\documentclass[a4paper,11pt]{article}
\usepackage{pos}

\usepackage{amsmath,amsfonts,slashed,empheq,tensor}

\def\one{\mbox{1 \kern-.59em {\rm l}}}

\def\R{\mathbb{R}}

\def\C{\mathbb{C}}

\def\cM{{\cal M}}

\def\cC{{\cal C}}

\def\nn{\nonumber}
\def\bea{\begin{eqnarray}}
\def\eea{\end{eqnarray}}
\def\be{\begin{equation}}
\def\ee{\end{equation}}

\newtheorem{thm}{Theorem}[section]
\newtheorem{lem}[thm]{Lemma}

\newcommand{\eq}[1]{(\ref{#1})}
\def\Tr{{\rm Tr}}

\newcommand{\del}{\partial}

\def\a{\alpha}          
\def\b{\beta}         
\def\g{\gamma}  
\def\d{\delta}

\def\l{\lambda} 
\def\L{\Lambda} 

\def\r{\rho}
 \def\s{\sigma}  
\def\t{\tau}

\newcommand{\End}{\rm{End}}

\def\hs{\mathfrak{hs}}
\def\NC{{\rm NC}}

\def\cA{{\cal A}} \def\cB{{\cal B}} \def\cC{{\cal C}}
\def\cD{{\cal D}}  
 \def\cH{{\cal H}} 
 \def\cK{{\cal K}} 
\def\cM{{\cal M}}  
  \def\cR{{\cal R}}
 \def\cT{{\cal T}} 
  
 \def\cZ{{\cal Z}}

\def\mso{\mathfrak{so}}

\def\del{\partial}

\title{Classical space-time geometry in the IKKT matrix model}

\author*{Harold C. Steinacker}

\affiliation{Department of Physics, University of Vienna,\\
  Boltzmanngasse 5, A-1090 Vienna, Austria}

\emailAdd{harold.steinacker@univie.ac.at}

\abstract{We discuss the reconstruction of generic 3+1-dimensional space-time geometries
from covariant quantum spaces as backgrounds in the IKKT matrix model.
An explicit recipe to realize generic classical  geometries is provided.
Even though this typically entails some  higher-spin contributions,
these do not significantly modify the physical content of the model in the weak gravity regime.
This justifies the framework for emergent gravity given by the
semi-classical matrix model, supplemented 
by an induced Einstein-Hilbert action which arises in the presence of fuzzy extra dimensions. 
}

\FullConference{%
  Corfu Summer Institute 2021 "School and Workshops on Elementary Particle Physics and Gravity"\\
  29 August - 9 October 2021\\
  Corfu, Greece
}

\begin{document}
\maketitle

\section{Introduction}

The purpose of these notes is two-fold. The first and main purpose is 
to provide a  justification for the geometric framework which is
underlying the higher-spin  gravity and gauge theory in the IKKT matrix model,
as described in a series of recent papers
\cite{Sperling:2019xar,Steinacker:2020xph,Fredenhagen:2021bnw,Asano:2021phy,Steinacker:2021yxt}. We will show  that
generic 3+1-dimensional space-time geometries can indeed be realized
as backgrounds within the IKKT matrix model, whose structure is that 
of covariant quantum spaces. This means that there is no explicit
Poisson tensor or $B$ field on space-time which would manifestly
break Lorentz invariance.

The second purpose of these notes is to summarize and discuss   
some further implications of emergent gravity in this framework, in particular the recent 1-loop computation
leading to the Einstein-Hilbert action \cite{Steinacker:2021yxt}. The underlying framework is 
now fully justified by the present  reconstruction of generic geometries.

The main result of the paper is a
recipe how to realize or reconstruct generic background geometries (with trivial topology) in the matrix model, starting from some metric $G_{\mu\nu}$
on space-time. Even though this was assumed in the above works, no full
justification has been given, and the statement is  in fact rather
 subtle. Here we show that generic classical 
geometries can be reconstructed via suitable matrix model backgrounds, 
provided we 
restrict ourselves to the weak gravity regime.
This means that the gravitational curvature length scale should be 
much larger than any other physical scale. Under this assumption, 
the reconstructed geometries are well approximated by their classical counterparts,
and can be described locally in terms of linearized perturbations of flat geometry.

The matrix models under consideration have an extremely simple structure, given by 
\begin{align}
S_{YM} = {\rm Tr} [T^{\dot a},T^{\dot b}][T^{{\dot a}'},T^{{\dot b}'}] \eta_{{\dot a}{\dot a}'} \eta_{{\dot b}{\dot b}'} \,\, + {\rm fermions} .
\label{MM-action}
\end{align}
Here $T^{\dot a},\, {\dot a}= 0,...,D-1$ are a set of hermitian matrices which transform under
a global $SO(D-1,1)$ symmetry acting on the
dotted Latin indices, and $\eta_{{\dot a}{\dot b}}$
can be interpreted as $SO(D-1,1)$- invariant metric on target space $\R^{D-1,1}$.
The models are invariant under gauge transformations
\begin{align}
 T^{\dot a} \to U^{-1} T^{\dot a} U \ .
\end{align}
It is straightforward to include fermions, which is very important for the 
quantization; in fact we will require 
maximal supersymmetry, as realized in the IKKT model \cite{Ishibashi:1996xs} with $D=10$.
There is no a priori notion of  space-time or
differential geometry;
all  geometrical structures relevant for the fluctuations on some given background solution
emerge dynamically within the model.
 We will show how generic 3+1-dimensional space-time geometries as required for
gravity  can be realized
as deformations of the covariant cosmic background  $\bar\cM^{3,1}$
introduced in \cite{Sperling:2019xar}.


A general framework which allows to make geometric sense of the matrix model is that
of quantized symplectic spaces. We consider any given set of matrices $T^{\dot a}$ as a
{\bf matrix configuration}. Since the action is given by the square of commutators, only ``almost-commutative'' matrix configurations are expected to play a significant role at low energies, i.e. matrices
whose commutators are much smaller in some sense than the matrices $T^{\dot a}$.
One can then argue on rather general grounds \cite{Steinacker:2020nva,Ishiki:2015saa}
that such matrix configurations can be interpreted in terms of a
quantized symplectic space $(\cM,\omega)$, where the algebra of functions $\cC(\cM)$
is replaced by the operator algebra $\End(\cH)$. More precisely, this is expected to hold
for some subspace of IR functions and almost-local operators; more details can be found in \cite{Steinacker:2020nva}. Such functions
\begin{align}
 \Phi\in\End(\cH) \sim \phi \in \cC(\cM)
\end{align}
can be identified with their classical
counterpart via some (de-) quantization map defined via quasi-coherent states.
We will work mostly in the semi-classical regime indicated by $\sim$, where commutators can be replaced by
Poisson brackets
\begin{align}
 [\Phi,\Psi] \sim i \{\phi,\psi\}
\end{align}
as familiar from quantum mechanics.
In particular, the $T^{\dot a}$ can accordingly
be viewed as quantized functions on $\cM$, which thereby define an embedding of $\cM$
into target space:
\begin{align}
 T^{\dot a} \sim t^{\dot a}: \quad \cM \hookrightarrow \R^{9,1} \ .
\end{align}
This suggests to interpret $\cM$  as a brane, very much like in string theory.
However from the point of view of the physics on $\cM$, the $T^{\dot a}$ and their commutators
\begin{align}
 \Theta^{\dot a\dot b} := i[T^{\dot a},T^{\dot b}] \ \sim - \{T^{\dot a},T^{\dot b}\}
\end{align}
play also another role, and can be related to geometric i.e. tensorial objects on $\cM$.

The key to understand $T^{\dot a}$ and $ \Theta^{\dot a\dot b}$ is to observe that they generate
  {\em Hamiltonian vector fields}  on $\cM$:
 \begin{align}
 E^{\dot a} [\phi] &:= \{T^{\dot a},\phi\}  \label{frame-general-0}\\
 \tensor{\cT}{^{\dot a}^{\dot b}} [\phi] &:= \{\Theta^{\dot a\dot b},\phi\}
 \label{torsion-general-0}
\end{align} 
acting on some test-function $\phi\in\cC(\cM)$.
These vector fields can be made more explicit by
introducing local coordinates $y^\mu$ on the $n$-dimensional manifold $\cM$.
Define 
 \begin{align}
 \tensor{E}{^{\dot a}^\mu} &:= \{T^{\dot a},y^\mu\}  \ ,
 \label{frame-general}\\
 \tensor{\cT}{^{\dot a}^{\dot b}^\mu} &:= \{\Theta^{\dot a\dot b},y^\mu\} \ ;
 \label{torsion-general}
\end{align} 
their significance will be clarified shortly.
We must carefully distinguish the different types of indices:
Greek indices $\mu,\nu =1,...,n$ will denote local coordinate indices on $\cM$,
which  play the role of
tensor indices. Dotted Latin indices $\dot a,\dot b =0,...,9$
 indicate frame-like indices which
are unaffected by a change of coordinates $y^\mu$, but transform under the global
$SO(1,9)$ symmetry  of the matrix model.
These frame-like indices will be raised and lowered with $\eta_{\dot a \dot b}$.
In particular, the $\tensor{E}{^{\dot a}^\mu}$ define vector fields
\begin{align}
 \tensor{E}{^{\dot a}} = \tensor{E}{^{\dot a}^\mu}\del_\mu
\end{align}
on $\cM$,
which  play a role of a (generalized)
{\em frame} on $\cM$.
This will  allow to understand the effective geometry and the gauge theory which arises on
$\cM$ through the matrix model.
In particular,  we can recognize
the infinitesimal gauge transformations in the matrix model
\begin{align}
 \d_\L T^{\dot a} = [T^{\dot a},\L] \sim i  \{T^{\dot a},\L\} = i \tensor{E}{^{\dot a}^\mu} \del_\mu\L
 \label{gaugetrafo-infinites}
\end{align}
 as generators of a sub-sector of diffeomorphisms on $\cM$, namely of the symplectomorphisms.
Finally, the tensor $\tensor{\cT}{^{\dot a}^{\dot b}^\mu}$ can be recognized as torsion of
the Weitzenb\"ock connection associated to the frame $\tensor{E}{^{\dot a}}$,
which is very useful to describe the non-linear regime of the
matrix model in the semi-classical regime \cite{Steinacker:2020xph,Langmann:2001yr}.

\paragraph{Covariant quantum space-time.}

In the following we will  focus on branes $\cM$ which are embedded in target space along the
 $\dot a,\dot b =0,...,3$ directions. Then the extra dotted indices will mostly be ignored, but they
 play a role once  fuzzy extra dimensions are included.
However, this assumption does not mean that $\cM$ is a 4-dimensional manifold;
if $\cM$ is 4-dimensional, then the Poisson tensor $\theta^{\mu\nu}$ on $\cM$  plays the role of some
background tensor on space-time, which is problematic since it breaks Lorentz invariance.
To avoid this we will consider a different class of {\bf covariant quantum spaces},
which have the structure of a $S^2$ bundle over space or space-time
\begin{align}
 \cM \cong S^2 \times \cM^{3,1} \qquad \mbox{locally} \ .
\end{align}
The prototype $\bar \cM$ of such a structure \cite{Sperling:2019xar}
is obtained as a certain projection of the fuzzy hyperboloid $H^4_n$ \cite{Sperling:2018xrm,Hasebe:2012mz},
and gives rise to a quantum space-time $\bar\cM^{3,1}_n$
with FLRW geometry and  Minkowski signature.
For other examples and approaches to covariant quantum spaces\footnote{The framework of \cite{Hanada:2005vr} is also somewhat similar to ours, but the bundles under consideration there are vastly bigger.}
see e.g. \cite{Ramgoolam:2001zx,Abe:2002in,Medina:2002pc,Grosse:2010tm,Hasebe:2012mz,Heckman:2014xha,Buric:2017yes,Manolakos:2019fle}.

Let us describe the structure of the covariant quantum space-time
$\bar \cM$  in some detail.
In the semi-classical limit $n\to\infty$, $\bar \cM$
reduces to an $SO(3,1)$- equivariant $S^2$ bundle over
$ \bar\cM^{3,1}$. 
The functions on the 6-dimensional $\bar\cM$ are
generated by  generators $x^\mu$ which describe $\bar\cM^{3,1}$, and
 $t_\mu$ which generate the internal sphere $S^2$.
 Both sets of generators transform covariantly under $SO(3,1)$, and satisfy the constraints
\begin{align}
 x_\mu x^\mu &= -R^2 - x_4^2 = -R^2 \cosh^2(\eta) \, ,
 \qquad R \sim \frac{r}{2}n   \label{radial-constraint}\\
 t_{\mu} t^{\mu}  &=  r^{-2}\, \cosh^2(\eta) \, \label{tt-constraint}\\
 t_\mu x^\mu &= 0 \ 
 \label{x-t-orth}
\end{align}
where indices are contracted with $\eta^{\mu\nu}$. Here
 $\eta \in (-\infty,\infty)$ plays the role of a FLRW time parameter,
featuring a big bounce at $\eta=0$.
The space of functions
decomposes into a direct sum $\End(\cH_n) = \oplus\, \cC^s$ of higher spin ($\hs$) modes on
$\cM^{3,1}$, which in the
semi-classical regime can be organized in terms of totally symmetric traceless tensors
\begin{align}
 \phi^{(s)} &= \phi_{\mu_1 ... \mu_s}(x) t^{\mu_1} ... t^{\mu_s} \nn\\
 \qquad \phi_{\mu_1 ... \mu_s} x^{\mu_i} &= 0\ = \phi_{\mu_1 ... \mu_s} \eta^{\mu_i \mu_j} .
 \label{Cs-explicit}
\end{align}
 $\bar\cM$ is
a symplectic manifold (which is quantized in the matrix model),
and the Poisson tensor $\theta^{\mu\nu} = \{x^\mu,x^\nu\}$
vanishes
upon projection to space-time $\cM^{3,1}$. This projection or averaging over $S^2$
will be denoted by $[.]_0$:
\begin{align}
 [\theta^{\mu\nu}]_0 \equiv \int\limits_{S^2} \theta^{\mu\nu} = 0 \ .
\end{align}
The more generic covariant quantum spaces under consideration here are by definition
the {\em same symplectic bundle} $\cM \cong \bar\cM$ , realized as a background of the model
through a different, perturbed
embedding map $T^{\dot a} \sim t^{\dot  a}$.
More explicitly,
\begin{align}
 T^{\dot a} = \bar T^{\dot a} + \cA^{\dot a}    \sim t^{\dot a} + \cA^{\dot a} \
\end{align}
where $\cA^{\dot a}$ are functions on $\cM$ or equivalently $\hs$-valued functions on $\cM^{3,1}$
which can be expanded in the form \eq{Cs-explicit}.
In particular, all these backgrounds are equivalent as symplectic spaces, and we will always
use the standard coordinate functions $x^\mu$ and $t^\mu$ as for the undeformed background $\bar\cM$,
with the same the symplectic form or Poisson structure.
This is very natural since symplectic manifolds are rigid, so that any deformation is
 equivalent (locally, at least) to the undeformed space by some diffeomorphism.

The purpose of this short paper is to clarify if and under what conditions
the  higher-spin gauge theory on $\cM^{3,1}$ can be reduced to (or is dominated by) the classical geometry
i.e. the lowest spin sector on $\cM^{3,1}$, which is supposed to play the role of physical space-time.
More explicitly, we want to understand if it is consistent to restrict to
fluctuations of the form
\begin{align}
 \cA^{\dot a} = \cA^{\dot a \mu}(x)\, t_\mu \ ,
\end{align}
dropping or neglecting higher-spin $\hs$ contributions $\cA^{\dot a \mu_1...\mu_s}\, t_{\mu_1} ... t_{\mu_s}$.
We will indeed establish that backgrounds of the structure
\begin{align}
 T^{\dot a} = T^{\dot a \mu}(x)\, t_\mu
\end{align}
are sufficiently rich to describe generic 3+1-dimensional space-time geometries, and
provide a self-consistent class of configurations in the matrix model where higher-spin corrections are
negligible in the {\bf weak gravity regime}, to be discussed below.

\section{Effective metric and frame on covariant quantum space-time}
\label{sec:eff-frame}

Now we establish the interpretation of $\tensor{E}{^{\dot a}^\mu}$ as frame on $\cM^{3,1}$.
As in any field theory, the effective metric governing some field or 
fluctuation mode is encoded in the kinetic term of the action.
Consider a
matrix background corresponding to some $2n$-dimensional brane $\cM\hookrightarrow \R^{3,1}\subset \R^{9,1}$.
Then the  kinetic (=quadratic) term for transversal fluctuations\footnote{The case of tangential fluctuations can be analyzed similar and leads to the same metric.} 
in Yang-Mills matrix models has the structure
\begin{align}
 S[\phi] =  \Tr([T^{\dot a},\phi][T_{\dot a},\phi]) 
  &\sim -\frac{1}{(2\pi)^n} \int\limits_\cM\Omega\;\{T^{\dot a},\phi\}\{T_{\dot a},\phi\} \nn\\
  &= - \frac{1}{(2\pi)^n} \int\limits_\cM\Omega\; \eta_{{{\dot a}}{\dot b}} \tensor{E}{^{\dot a}^\mu} \tensor{E}{^{\dot b}^\nu}\del_\mu\phi\del_\nu\phi  \nn\\
  &= -\frac{1}{(2\pi)^n} \int\limits_\cM\Omega\; \g^{\mu\nu}\del_\mu\phi\del_\nu\phi
  \label{S-kinetic-1}
 \end{align}
in the semi-classical regime, recognizing \eq{frame-general}.
Here $\Omega$ is the symplectic volume form on $\cM$, and
\begin{align}
 \g^{\mu\nu} := \eta_{{{\dot a}}{\dot b}} \tensor{E}{^{\dot a}^\mu} \tensor{E}{^{\dot b}^\nu} \ .
 \label{metric-gamma}
\end{align}
This is clearly the metric determined by the frame $\tensor{E}{^{\dot a}^\mu}$;
however the effective metric acquires an extra conformal factor, which arises as follows.
In the case of covariant quantum spaces under consideration, we can assume that
$\cM = \bar\cM = S^2 \times \cM^{3,1}$, with a global $SO(3)$ symmetry acting on $S^2$ and $\cM^{3,1}$
simultaneously. Then
 $\Omega$ factorizes into the volume of the $S^2$ fiber
times the effective density $\r_M$ on space-time $\cM^{3,1}$ \cite{Sperling:2019xar}:
\begin{align}
 \Omega = \r_M  d^4 x\, \Omega_2 \ , \qquad \r_M = \frac 1{r^2 R^2\sinh(\eta)} 
  \sim L_{\rm NC}^{-4} \ .
  \label{sympl-density}
\end{align}
Here $S^2$ is normalized with volume $4\pi$, and $x^\mu$ are the Cartesian coordinates  \eq{radial-constraint} on $\cM^{3,1}$ or\footnote{Recall that $\cM = \bar\cM$ as a manifold, only the embedding and the frame are deformed.} $\bar \cM^{3,1}$.
$L_{\rm NC}$ characterizes the scale of noncommutativity.
Then \eq{S-kinetic-1}  can be written in a more familiar form
 \begin{align} 
   S[\phi]
  &\sim -\frac{1}{2\pi^2} \int\limits_\cM\r_M\; \g^{\mu\nu}\del_\mu\phi\del_\nu\phi 
  = -\frac{1}{2\pi^2}\int\limits_{\cM^{3,1}} d^{4}x\, \sqrt{|G_{\mu\nu}|} G^{\mu\nu}\del_\mu\phi\del_\nu\phi \ .
\end{align}
We can now read off the  {\bf effective metric} on $\cM^{3,1}$:
\begin{align}
 G^{\mu\nu} &= \r^{-2}\, \g^{\mu\nu} \ 
  \label{eff-metric-def}
\end{align}
where $\r$ is the {\bf dilaton}, which relates the symplectic density $\r_M$
to the Riemannian density via
\begin{align}
  \r^{-2}\sqrt{|G_{\mu\nu}|} = \r_M \  = \r^{2} \sqrt{|\g_{\mu\nu}|}
  \label{rho-metric-general-cov}
\end{align}
using $\sqrt{|G_{\mu\nu}|} =\r^{4} \sqrt{|\g_{\mu\nu}|}$.
From the string theory point of view, the metric  $G_{\mu\nu}$ can be interpreted as open-string metric on $\cM^{3,1}$.
Noting that
\begin{align}
 \sqrt{|\g^{\mu\nu}|} = |\det\tensor{E}{^{\dot a}^\mu}| \ ,
\end{align}
 the dilaton is determined by the frame as
\begin{align}
 \r^{2}  =  \r_M\, |\det\tensor{E}{^{\dot a}^\mu}| \ .
 \label{dilaton-eq-3}
\end{align}
It is important that the frame $\tensor{E}{^{\dot a}^\mu}$
in the present context does {\em not} admit local $SO(3,1)$ gauge transformations
acting on $\dot a$, only global $SO(3,1)$ transformations are allowed.
The frame is a physical object here which is subject to certain constraints \eq{div-free-frame},
and determines not only the metric but also  additional physical information,
such as the dilaton $\r$ and also an axion $\tilde\r$ \eq{T-del-onshell}.

\subsection{Cosmological FLRW solution}
\label{sec:cosm-background}

A special case of the above class of backgrounds is given by
 \begin{align}
 T^\mu = \frac{1}{R} M^{\mu 4} \ \sim t^\mu
 \label{T-background-def}
\end{align}
where $M^{ab}$ are  generators of the doubleton representation $\cH_n$
of $\mso(4,1) \subset \mso(4,2)$.
It is easy to see that $T^\mu$ is a solution of the matrix model in the presence of a suitable mass term; we shall simply discuss some of its properties here.
 $T^\mu$ defines a matrix configuration with manifest $SO(3,1)$ symmetry, which
in the semi-classical regime
reduces to a 6-dimensional background $\bar\cM$ which is an $S^2$ bundle over $\bar\cM^{3,1}$.
The Cartesian coordinate functions on the base manifold $\bar\cM^{3,1}$ arise as
 \begin{align}
 X^\mu = r\, M^{\mu 5} \ \sim x^\mu \ .
 \label{X-background-def}
\end{align}
We will focus on the semi-classical (Poisson)
limit $n \to \infty$, working with  commutative functions of $x^\mu$ and $t^\mu$,
but keeping the Poisson structure  $[.,.] \sim i \{.,.\}$.
Then $\End(\cH_n) \sim \cC$ reduces to the
algebra of functions on the bundle space $\cM \cong \C P^{2,1}$, dropping the bar for now.
The sub-algebra $\cC^0 \subset \cC$ of functions on the base space $\cM^{3,1}$ is  generated by the
\begin{align}
 x^\mu:\, \cM^{3,1}\hookrightarrow \R^{3,1}
\end{align}
for $\mu=0,...,3$, which are interpreted as Cartesian
coordinate functions.
The generators $x^\mu$ and $t^\mu$ satisfy
the  constraints  \eq{x-t-orth},
which arise from the special properties of $\cH_n$.
The $t^\mu$ generators  describe the $S^2$ fiber over $\cM^{3,1}$, which is
 space-like due to \eq{x-t-orth}. Here
$\eta$ plays the role of a time parameter, defined via
\begin{align}
 x^4 = R \sinh(\eta) \ .
 \label{x4-eta-def}
\end{align}
Hence $\eta = const$ defines a foliation of $\cM^{3,1}$ into space-like surfaces $H^3$; this 
can be related to the scale parameter of a FLRW cosmology with $k=-1$.
Note that $\eta$ runs from $-\infty$ to $\infty$, and
the sign of $\eta$ distinguishes the two degenerate sheets of $\cM^{3,1}$ linked by a Big Bounce,
cf. \cite{Steinacker:2019fcb}.
The Poisson brackets on $\bar\cM$ are given explicitly by
\begin{align}
 \{x^\mu,x^\nu\} &= \theta^{\mu\nu}  = - r^2 R^2 \{t^\mu,t^\nu\} \ ,  \nn\\
 \{t^\mu,x^\nu\} &= \frac{x^4}{R} \eta^{\mu\nu} \ ,
 \label{Poisson-brackets-M31} \
\end{align}
where the Poisson tensor $\theta^{\mu\nu}$ satisfies the constraints
\begin{subequations}
\label{geometry-H-theta}
\begin{align}
 t_\mu \theta^{\mu\a} &= - \sinh(\eta) x^\a , \\
 x_\mu \theta^{\mu\a} &= - r^2 R^2 \sinh(\eta) t^\a , \label{x-theta-contract}\\
 \eta_{\mu\nu}\theta^{\mu\a} \theta^{\nu\b} &= R^2 r^2 \eta^{\a\b} - R^2 r^4 
t^\a t^\b + r^2 x^\a x^\b  \ .
%
\end{align}
\end{subequations}
$\theta^{\mu\nu}$ can be expressed in terms of $t^\mu$ as
\begin{align}
 \theta^{\mu\nu} &= \frac{r^2}{\cosh^2(\eta)} 
   \Big(\sinh(\eta) (x^\mu t^\nu - x^\nu t^\mu) +  \epsilon^{\mu\nu\a\b} x_\a t_\b \Big)
 \label{theta-P-relation} \ ,
\end{align}
and can therefore be viewed as spin $1$ valued ``function'' on $\cM^{3,1}$.
More generally, the space of functions $\cC$ on $\cM$ decomposes into a tower of higher-spin ($\hs$)
valued functions
\begin{align}
 \cC = \bigoplus_{s\geq 0} \cC^s
\end{align}
on  $\cM^{3,1}$, where  $\cC^s$ is spanned  by irreducible polynomials \eq{Cs-explicit} of degree $s$ in $t^\mu$.
The Poisson brackets do not respect the decomposition into $\cC^s$, but the following
holds
\begin{align}
\{\cC^s, x^\mu\} \in \cC^{s+1} \oplus  \cC^{s-1} \ 
\label{poisson-x-decomp}
\end{align}
noting that $\theta^{\mu\nu} \in\cC^1$.

\paragraph{Frame,  metric and torsion on $\bar\cM^{3,1}$.}

Following the general strategy discussed above,
we can extract the effective metric on $\bar\cM^{3,1}$.
Frame and metric are obtained in Cartesian coordinates from \eq{Poisson-brackets-M31} as
\begin{align}
 E^{\dot a}  &= \{t^{\dot a},\cdot\} = E^{\dot a\mu} \del_\mu ,
  \qquad  E^{\dot a\mu} = \eta^{\dot a\mu} \sinh(\eta)\ , \nn\\[1ex]
 \g^{\mu\nu} &= \eta_{\dot a\dot b}E^{\dot a\mu}E^{\dot b\nu}
  = \sinh^2(\eta) \eta^{\mu\nu}   \ .
 \label{gamma-vielbein}
\end{align}
Recalling that $\r_M \sim \sinh(\eta)^{-1}$,
the effective metric on $\bar\cM^{3,1}$ and the dilaton are  obtained as
\begin{align}
   G_{\mu\nu} &= \sinh^{3}(\eta)  \g_{\mu\nu} \ = \sinh(\eta) \eta_{\mu\nu}  \ , \nn\\ 
   \rho^2 &= \sinh^{3}(\eta) \ .
   \label{eff-metric-Gbar}
\end{align}
This metric is $SO(3,1)$-invariant with signature $(-+++)$ and
conformal to the induced (``closed-string'') metric $\eta_{\mu\nu}$.
It can be written in standard FLRW form as follows \cite{Sperling:2019xar}
\begin{align}
 d s^2_G = G_{\mu\nu} d x^\mu d x^\nu
   &= -d t^2 + a^2(t)d\Sigma^2 \,
   \label{eff-metric-FRW-2}
\end{align}
where  $d\Sigma^2$ is the metric on $H^3$, and the FLRW time $t$ is related to the time parameter $\eta$ via
\begin{align}
 a(t) \sim R \sinh^{3/2}(\eta) =:  L_{\rm cosm}, \qquad t \to \infty \ .
 \label{L-cosm}
\end{align}
One finds $a(t) \sim \frac 32 t$ for late times, and $a(t) \sim t^{1/5}$ near the Big Bounce.
The torsion tensor \eq{torsion-general} is also easily computed
using $\Theta^{\dot a\dot b} = \frac{1}{R^2} \cM^{\dot a\dot b}$,
which gives
\begin{align}
 \tensor{\cT}{^{\dot a}^{\dot b}^\mu} &= \{\Theta^{\dot a\dot b},x^\mu\}
  =  \frac{1}{R^2}\big(\eta^{\dot a\mu}x^{\dot b} -    \eta^{\dot b\mu}x^{\dot a}\big) \
\end{align}
in Cartesian coordinates $x^\mu$. This can be recast as a rank 3 tensor
on $\cM^{3,1}$ using the frame $E^{\dot a\mu}$,
 \begin{align}
 \tensor{\cT}{_\nu_\s^\mu} &= \frac{1}{R^2 \r^2} \big(\d_\nu^{\mu} \t_\s - \d_\s^{\mu}\t_\nu\big) \
 \label{torsion-M31-frame}
\end{align}
where
\begin{align}
 \t_\mu = G_{\mu\nu}\t^\nu =  G_{\mu\nu} x^\nu = \sinh(\eta)\eta_{\mu\nu }x^\nu
  \label{torsion-M31-coords}
\end{align}
is a global time-like $SO(3,1)$-invariant vector field on the FLRW background.

\paragraph{Late-time regime and noncommutativity scale.}

Consider the regime of late time or large $\eta$, so that $\sinh(\eta) \gg 1$.
Then the Poisson tensor $\theta^{\mu\nu}$ \eq{theta-P-relation} 
reduces to
\begin{align}
 \theta^{\mu\nu}
\sim \frac{r^2}{\cosh(\eta)} (x^\mu t^\nu - x^\nu t^\mu) \ ,\qquad\eta\to\infty .
   \label{theta-approx}
\end{align}
More specifically, consider
 some given reference point $\xi = (x^0,0,0,0)$ on $\cM$. Then
this reduces to 
\begin{align}
  \theta^{0i} &\stackrel{\xi}{=}  \frac{r^2}{\cosh^2(\eta)} \sinh(\eta) x^0 t^i 
   \ \sim \ r^2 R t^i  \quad = O(L^2_{\rm NC})  \nn\\
  \theta^{ij} &\stackrel{\xi}{=} \frac{r^2}{\cosh^2(\eta)} x^0 \epsilon^{0ijk} t_k 
  \ \sim \ \frac 1{\sinh(\eta)} r^2 R \epsilon^{ijk} t^k \quad = O(r R) \ ,
 \label{CR-explicit-ref}
\end{align}
where 
\begin{align}
 L^2_{\rm NC} = R r \cosh(\eta) \
 \label{L-NC-M31}
\end{align}
is the effective scale of noncommutativity on $\cM^{3,1}$ (cf. \eq{sympl-density}),
using $|t| \sim r^{-1}\cosh(\eta)$ \eq{tt-constraint}.
Even though this grows with $\eta$, it is much shorter than the cosmic curvature 
scale \eq{L-cosm}:
\begin{align}
 \frac{L_{\rm cosm}^2}{L_{\NC}^2} \sim \frac{R}{r} \cosh^2(\eta) 
 \sim n \cosh^2(\eta)\ .
 \label{scales-hierarchy-cNC}
\end{align}
Therefore there is plenty of space for interesting physics in between.
In particular, $ \theta^{0i} \sim r^2 R t^i  \gg \theta^{ij} $  at late times $\eta \gg 1$.
The space-like generators $t^i$ 
describe the internal fuzzy sphere $S^2_n$ with
\begin{align}
 \{t^i,t^j\} &\stackrel{\xi}{=} -\frac{1}{r^2 R^2} \theta^{ij} = -\frac 1{R\sinh(\eta)} \epsilon^{ijk} t^k \
 \label{t-t-CR}
\end{align}
and generate the higher-spin algebra $\hs$.
Even though $t^0\stackrel{\xi}{=} 0$ vanishes as function at $\xi$, it
is  a non-trivial generator which
induces local time translations via $\{t^0,.\}$.

\subsection{Derivations}

\paragraph{Fuzzy hyperboloid $H^4_n$.}

The above space-time $\cM^{3,1}$ can be understood as a projection of the 
fuzzy hyperboloid $H^4_n$ \cite{Sperling:2018xrm}, which can be viewed as a submanifold of
$\R^{4,1}$ defined in terms of the 5 generators
 \begin{align}
 X^a = r\, M^{a 5} \ \sim x^a \ , \qquad a=0,...,4
 \label{X-background-def-H4}
\end{align}
(cf. \eq{X-background-def}) which transform as vectors of $SO(4,1)$.
The underlying symplectic space is the same as for $\cM^{3,1}$, given by 
the non-compact projective space $\C P^{2,1}$ which is nothing but 
(projective) twistor space, cf. \cite{Steinacker:2022jjv}.
The Poisson structure on the bundle space allows to define derivations 
as follows
\begin{align}
  \eth^a \phi \coloneqq -\frac{1}{r^2 R^2}\theta^{ab}\{x_b,\phi\}
  = \frac{1}{r^2 R^2}x_b \{\theta^{ab} ,\phi\}, \qquad
\phi\in\cC \; .
\label{eth-def}
\end{align}
They satisfy the useful identities
\begin{align}
x^a \eth_a \phi &= 0 \ , \nn\\  
\eth^a x^c &=  \eta^{ab} + \frac{1}{R^2}x^a x^b \ , \nn\\
 \eth^a (\{x_a,\phi\}) &= 0 \  \label{eth-bracket-id}
\end{align}
for any $\phi\in\cC$. Furthermore, we note that all (even $\hs$-valued) 
Hamiltonian vector fields
on $H^4_n$ are  tangential to $H^4\subset \R^{4,1}$, due to the identity
\begin{align}
 x^a\{x_a,\L\} = 0 \ .
\end{align}

\paragraph{Derivatives on $\cM^{3,1}$.}

Since the algebra of functions $\cC$ for $\cM^{3,1}$ and $H^4_n$ is the same,
we can use the above  derivative operators to define the following derivations on  $\cM^{3,1}$
\begin{align}
 \del_\mu &:=  \eth_\mu - x_\mu\frac 1{x_4} \eth_4 \qquad \mbox{on}\quad \cC  \ .
 \label{eth-M31-def}
\end{align}
Using the identities \eq{eth-bracket-id}, it is easy to show 
\begin{align}
 \del_\mu x^\nu &= \d_\mu^\nu \nn\\
 \del_\mu(\r_M\theta^{\mu\nu}) &= 0 \ .
 \label{del-theta-Poisson-M31}
\end{align}
This will imply that all Hamiltonian vector fields  on $\cM^{3,1}$, in particular the frame, are conserved.

\section{Divergence-free vector fields on $H^4$ and $\cM^{3,1}$}

Divergence-free vector fields will play an important role in the following.
Clearly
any vector field $V^a$  on $H^4$ can be mapped to a vector field $V^\mu$ on $\cM^{3,1}$, by
simply dropping the $V^4$ component (in Cartesian coordinates). 
This can be understood as  push-forward via a
projection \cite{Sperling:2019xar}. For example, a Hamiltonian vector field $V^a = \{T,x^a\}$
is mapped to $V^\mu = \{T,x^\mu\}$ in Cartesian coordinates.
Conversely, any vector field $V^\mu$ on $\cM^{3,1}$ can be lifted to $H^4$ by defining 
\begin{align}
V^4 := - \frac{1}{x_4} x_\mu V^\mu \ ,
\label{V4-def}
\end{align}
which  defines a tangential vector field $V^a x_a = 0$ on $H^4$.
We claim that this correspondence maps divergence-free vector fields $\eth_a V^a = 0$ on $H^4$
to divergence-free vector fields  on $\cM^{3,1}$, in the sense that
\begin{align}
 \del_\mu (\r_M V^\mu) = 0  \ .
\end{align}
Here $\r_M$ is the symplectic density \eq{sympl-density}  on $\cM^{3,1}$, which in Cartesian coordinates
is given by $\r_M = \sinh(\eta)^{-1}$.
In fact the following more general result holds:

\begin{lem}
 \label{lemma-div-M31}

 Let $V^a$ be a (tangential) vector field on $H^4$, i.e. $V^a x_a = 0$.
Then its reduction (or push-forward) $V^\mu$ to $\cM^{3,1}$ satisfies
 \begin{align}
  \eth_a V^a &= \sinh(\eta) \del_\mu(\r_M V^\mu)
   \label{div-H4-M31}
\end{align}
Conversely, the lift of $V^\mu$ to $H^4$ defined by \eq{V4-def}
satisfies \eq{div-H4-M31}.
If $V^a$ is  divergence-free on $H^4$ i.e. $\eth_a V^a = 0$, then
its reduction to $\cM^{3,1}$ satisfies
\begin{align}
 \del_\mu(\r_M V^\mu) =0 \ .
 \label{div-free-M31}
\end{align}
In particular, all  Hamiltonian vector fields on fuzzy $H^4$ and $\cM^{3,1}$ are conserved, in the sense
\begin{align}
 \eth_a \{x^a, T\} = 0, \qquad \del_\mu(\r_M \{x^\mu, T\})
 \label{conserv-Hamiltonian}
\end{align}

\end{lem}

\begin{proof}

Using the definition of $\del_\mu$ \eq{eth-M31-def} on $\cC$, we  compute
\begin{align}
 \eth_a V^a &=  \eth_\mu V^\mu  + \eth_4 V^4   \nn\\
  &= \big(\del_\mu + \frac 1{x_4} x_\mu\eth_4\big)  V^\mu
   + \eth_4 V^4    \nn\\
  &= \del_\mu V^\mu
  + \frac 1{x_4} \eth_4 (x_\mu V^\mu) - \frac 1{x_4} V^\mu\eth_4 x_\mu
  + \eth_4 V^4    \nn\\
  &= \del_\mu V^\mu
  - \frac 1{x_4} \eth_4 (x_4 V^4)
  -  \frac 1{R^2} x_\mu V^\mu
  + \eth_4 V^4    \nn\\
   &= \del_\mu V^\mu
   - \frac 1{x_4} V^4 \eth_4 x_4
  + \frac 1{R^2} x_4 V^4 \nn\\
   &= \del_\mu V^\mu - \frac 1{x_4} V^4  \nn\\
  &= \sinh(\eta) \del_\mu\big(\frac 1{\sinh(\eta)} V^\mu\big) \ .
\end{align}
\eq{conserv-Hamiltonian} now follows using  \eq{eth-bracket-id}.

\end{proof}
In particular, the identity
\eq{del-theta-Poisson-M31} can now be understood
by noting that $V^a = \{x^\nu,x^a\}$ is conserved on $H^4$.
We also note that
the divergence constraint \eq{div-free-M31} for vector fields on $\cM^{3,1}$ can be written using \eq{rho-metric-general-cov} in   covariant form in terms of the effective metric $G^{\mu\nu}$  on $\cM^{3,1}$:
\begin{align}
0 &= \nabla_\mu(\r^{-2} V^\mu )
  = \frac{1}{\sqrt{|G|}}\del_\mu\big(\r_M V^\mu\big) \
\end{align} 
where $\nabla$ is the Levi-Civita connection corresponding to $G$.

\section{Generic backgrounds from deformed $\cM^{3,1}$}

Starting from the above FLRW background, we can obtain more generic 
geometries as deformations, by simply adding fluctuations of the background:
\begin{align}
 T^a = \bar T^a + \cA^a 
\end{align}
The fluctuations $\cA^a$ are any $\hs$ valued gauge fields,
which are governed by a Yang-Mills gauge theory. 
We want to focus in the following on purely geometric deformations, leaving aside the
higher spin modes. We therefore focus on fluctuations of the form
\begin{align}
 \cA^a = \cA^{a\mu}(x) t_\mu 
\end{align}
Since we don't want to restrict ourselves to the linearized perturbations, we simply
consider generic backgrounds of the form 
\begin{align}
 T^a = T^{a\mu}(x) t_\mu 
 \label{general-BG}
\end{align}
which include the cosmic background for $T^{\dot a\mu}(x) = \eta^{\dot a\mu}$.
As discussed in section \ref{sec:eff-frame}, such a background defines a frame
\eq{general-BG} 
\begin{align}
  \tensor{E}{^{\dot a}^\mu} = \{T^{\dot a} ,x^\mu\}
  \label{frame-2}
\end{align}
Taking into account the above results, we conclude that any such frame
satisfies the divergence constraint \cite{Fredenhagen:2021bnw}
\begin{align}
\boxed{\
 \del_\mu\big(\r_M \tensor{E}{^{\dot a}^\mu}\big) = \nabla_\mu(\r^{-2} \tensor{E}{^{\dot a}^\mu}) \ .
\ }
\label{div-free-frame}
 \end{align}
In the following we will establish the converse statement: any frame given by 
divergence-free
vector fields can indeed be implemented as above, for a suitable background 
of the form \eq{general-BG}. Moreover, the $T^{\dot a}$  can be computed explicitly.
This entails in general some extra $\hs$ valued contribution
to the frame, which will be shown to be insignificant in section \ref{sec:weak-grav}
in the weak gravity regime.

\subsection{Reconstruction of divergence-free vector fields}
\label{sec:reconstruction-VF-M31}

We start by recalling two results given in \cite{Asano:2021phy}, starting with the Euclidean case:

\begin{lem}
 \label{lem:H4-reconstruct}
Given any divergence-free tangential vector field $\eth_a V^a = 0$ on $H^4$
with $V^a \in \cC^0$,
there is a unique generator $T\in\cC^1$ such that
\begin{align}
 V^a = \{T,x^a\}_0 \ .
 \label{VF-reconstruction-H4}
\end{align}
This $T$ is given explicitly by
\begin{align}
\boxed{\
 T :=   -3(\Box_H -4r^2)^{-1}\{V^a,x_a\} =: \cD^+(V)   \qquad \in \cC^1 \
 \ }
\end{align}
where $\Box_H = \{x^a,\{x_a,.\}\}$.

\end{lem}

However, the Hamiltonian vector field $\{T,x^a\}$ generated by the above $T \in \cC^1$
contains  in general also a spin 2 component 
\begin{align}
 V^{(2)a} := \{T,x^a\}_2 = \cD^{++}(V^a) \qquad \in \cC^2 \ 
 \label{vectorfield-C2-component}
\end{align}
 which is also divergence-free
$\eth_a V^{(2)a} = 0$. 
Here $\cD^{++}$ is an $SO(4,1)$ intertwiner given by 
\begin{align}
\cD^{++} (\cA) =  \{\cD^+(\cA),x^a\}_+   
  \label{Dpp}
\end{align}
which satisfies 
\begin{align}
  \int \cD^{++}(\cA^a) \cD^{++}(\cB^a)
  &\approx 4\int\cA^a \cB^a  \ , \label{Dpp-isometry}\\
  \cD^{++}(x^{a}\cA)  &\approx \  x^{a}(\cD^{++}\cA) \ , \nn\\ 
 \cD^{++}(\eth^{a}\cA)  &\approx \  \eth^{a}(\cD^{++}\cA) \ .
 \label{Dpp-deriv-approx}
\end{align}
in the regime where $r^{-2}\Box \gg 1$, i.e. not in the extreme IR regime.
Therefore the above reconstruction of vector fields on
$H^4$ generically leads to extra $\hs$ components  $V^{(2)a} \in \cC^2$, which however encode the same information as $V^a$.
It remains an open question if these can be cancelled by allowing higher-spin corrections
to the coordinate generators $x^a$.

We can use these results to obtain an analogous ``reconstruction'' statement on $\cM^{3,1}$ \cite{Asano:2021phy}:

\begin{lem}
 \label{lemma-reconstruct-M31}
Given any $\cC^0$-valued divergence-free vector field $V^\mu$ on $\cM^{3,1}$,
\begin{align}
 \del_\mu(\r_M V^\mu) = 0
\end{align}
there is a generating function $T \in\cC^1$ such that
\begin{align}
 V^\mu = \{T,x^\mu\}_0 \ .
 \label{T-generate-V}
\end{align}
Explicitly, $T$ is given by
\begin{align}
\boxed{\ 
 T = -3(\Box_H -4r^2)^{-1}\big(\{V^\mu,x_\mu\} + \{V^4,x_4\}   \big)
\ }
\end{align}
where  
\begin{align}
 V^4 = - \frac{1}{x_4} x_\mu V^\mu \ .
\end{align}

\end{lem}

This is simply obtained by lifting $V^\mu$ to a divergence-free vector field
$V^a$ on $H^4$
as in Lemma \ref{lemma-div-M31}. Then the  result \eq{VF-reconstruction-H4} on $H^4$
states that $V^a = \{T,x^a\}_0$ for some $T\in\cC^1$,
which implies $V^\mu = \{T,x^\mu\}_0$.
Moreover, $T$ is uniquely determined by \eq{T-generate-V}.
One can show that this spin 2 component vanishes only for $T \in \mso(4,1)$.

To summarize, we have shown that every divergence-free vector field on $\cM^{3,1}$
can be realized or reconstructed as Hamiltonian vector field, 
i.e. $V^\mu = \{T,x^\mu\}_0$. However, this
entails the presence of a spin two sibling 
$V^{(2)\mu} = \{T,x^\mu\}_2 \ \in \cC^2$. 
In other words, the Hamiltonian vector field generated by $T\in\cC^1$
acts on a  function $\phi = \phi(x)\in\cC^0$ via
\begin{align}
 \{T,\phi\} = \{T,x^\mu\}\del_\mu\phi = (V^\mu + V^{(2)\mu})\del_\mu\phi \ .
\end{align}
Both components of  
$V^\mu + V^{(2)\mu} \in \cC^0 \oplus \cC^2$ are  isomorphic
 as $\mso(4,1)$  modes. 
This applies in particular to the frame
in the effective field theory  on $\cM^{3,1}$ arising from matrix models.

\subsection{Reconstruction of classical geometry}
\label{sec:reconstruction-M31}

Now we apply the results of the previous section to reconstruct a classical frame 
$\tensor{e}{^{\dot a}^\mu}$
within the present framework. This is the basis for describing gravity 
through the effective metric on a suitable covariant quantum spaces.
It is clear from \eq{div-free-frame} that only divergence-free frames can be realized here,
but this does not restrict
the possible metrics as explained in section \ref{sec:generic-4D}.
Hence for any divergence-free classical frame $\tensor{e}{^{\dot a}^\mu}$, there is a
unique $T^{\dot a} \in \cC^1$ given by
\begin{align}
 T^{\dot a} &= -3(\Box_H -4r^2)^{-1}\big(\{\tensor{e}{^{\dot a}^\mu},x_\mu\} 
+ \{\tensor{e}{^{\dot a}^4},x_4\}  \big) \nn\\
 &= -3(\Box_H -4r^2)^{-1}\big(\{\tensor{e}{^{\dot a}^\mu},x_\mu\} 
 - \frac{1}{x_4}\{\tensor{e}{^{\dot a}^\mu} x_\mu,x_4\}  \big)
\end{align}
such that 
\begin{align}
 \tensor{e}{^{\dot a}^\mu} = \{T^{\dot a},x^\mu\}_0 \ .
 \label{frame-reconstruct-0}
\end{align}
E.g.~for the cosmic  frame
$\tensor{e}{^{\dot a}^\mu} = \sinh(\eta) \eta^{\dot a \mu}$ on $\cM^{3,1}$, this gives 
\begin{align}
 \tensor{e}{^{\dot a}^4} &= - \frac{x_\mu}{x_4} \tensor{e}{^{\dot a}^\mu} 
 =  - \frac 1r\, x^{\dot a} 
\end{align}
and we recover the background \eq{T-background-def}
\begin{align}
T^{\dot a} &= -3(\Box_H -4r^2)^{-1}\big(\{\tensor{e}{^{\dot a}^\mu},x_\mu\} 
 - \frac{1}{r}\{x^{\dot a},x_4\} \big)    \nn\\
 &= 6(\Box_H-4r^2)^{-1}\{\sinh(\eta),x^{\dot a}\} \nn\\
 &= t^{\dot a} \ 
\end{align}
using $\Box_H t^\mu = -2 r^2 t^\mu$.
The generator $T^{\dot a}\in\cC^1$ is uniquely determined by \eq{frame-2}.
However, the reconstructed frame will in general contain  higher spin $\hs$ components
$\{T^{\dot a},x^\mu\}_+ \in \cC^2$  due to \eq{poisson-x-decomp}.
Even though these drop out  in the linearized theory upon averaging over $S^2_n$,
this is no longer true in  the non-linear regime, and we must
clarify the importance of these contributions.

\section{Weak gravity regime  and  classical geometry}
\label{sec:weak-grav}

For covariant quantum spaces, the frame is in general higher-spin valued, 
and so is the metric. To describe real physics, we should be able to recover classical geometries in terms of backgrounds which contain no significant 
higher-spin contributions. This leads to the 
following problem: 
For any given classical divergence-free frame $\tensor{e}{_{\dot\a}^\mu} \in\cC^0$,
we would like to find generators $T_{\dot\a}$ 
such that the (generally $\hs$-valued) reconstructed frame
\begin{align}
\tensor{E}{_{\dot\a}^\mu} = \{T_{\dot\a},x^\mu\}  \
 \label{frame-reconstruct-C0-def}
\end{align}
reproduces the classical frame through its spin 0 component  $\cC^0$:
\begin{align}
 [\tensor{E}{_{\dot\a}^\mu}]_0 = [\{T_{\dot\a},x^\mu\}]_0 =   \ \tensor{e}{_{\dot\a}^\mu} \ .
 \label{frame-reconstruct-C0}
\end{align}
This problem of {\bf frame reconstruction} is of central importance,
since gravity requires to realize generic geometries in the matrix model framework. 
A solution of problem is given by the
results in the previous sections
as follows:
\begin{align}
 T_{\dot\a} &= \cD^+(\tensor{e}{_{\dot\a}}) = -3(\Box_H -4r^2)^{-1}\big(\{\tensor{e}{_{\dot\a}^\mu},x_\mu\} 
  + \frac{1}{x_4}\{x_4, x_\mu\tensor{e}{_{\dot\a}^\mu}\}   \big) \qquad \in \cC^1 \ .
  \label{frame-reconstruction}
\end{align}
However in general,
\begin{align}
\tensor{E}{_{\dot\a}^\mu} = \{T_{\dot\a},x^\mu\} 
= \tensor{e}{_{\dot\a}^\mu} + \cD^{++}(\tensor{e}{_{\dot\a}^\mu})
\ \in \ \cC^0 \oplus \cC^2
 \label{frame-reconstruct-C0-C2}
\end{align}
 contains also  $\cC^2$ contributions, which
 vanish only for very special ``pure backgrounds'':

\subsubsection{Pure backgrounds}

Consider the  class of {\bf pure backgrounds} $T_{\dot\a}\in\cC^1$,
which have the property that the frame is a pure function,
\begin{align}
  \tensor{E}{_{\dot\a}^\mu} &=  \{T_{\dot\a},x^\mu\}  \qquad   \in \cC^0 \ .
 \label{pure-BG}
\end{align}
It is not hard to see that all such backgrounds
are some linear combination of $\mso(4,1)$ generators
\begin{align}
 T_{\dot\a} &= \ e_{{\dot\a}; bc} \theta^{bc}\, \qquad\in \cC^1 \ .
\end{align}  
This leads to the $\cC^0$ -valued frame and torsion
\begin{align}  
\tensor{E}{_{\dot\a}^\mu} &=  e_{{\dot\a}; bc} \{\theta^{bc},x^\mu\}
  = -r^2  e_{{\dot\a}; bc}(\eta^{b\mu} x^c - \eta^{c\mu} x^b) \ ,  \nn\\
\tensor{T}{_{\dot\a}_{\dot\b}^\mu} &= - \{\{T_{\dot\a},T_{\dot\b}\},x^\mu\} 
 = r^2 c_{{\dot\a}{\dot\b};ab}(\eta^{a\mu} x^b - \eta^{b\mu} x^a)  \ 
 \label{pure-BG-explicit}
\end{align}
where $\{T_{\dot\a},T_{\dot\b}\} = c_{{\dot\a}{\dot\b};ab}\theta^{ab}$.
These configurations comprise the
cosmic background solution, which is recovered for 
\begin{align}
T_{\dot\a} = t_{\dot\a} =  \frac 1{2R}(\d_c^4 \eta_{{\dot\a} b} - \d_b^4 \eta_{{\dot\a} c}) \theta^{bc}\ .
\end{align}
In particular,  any classical frame of the form
\begin{align}
 \tensor{e}{_{\dot\a}^\mu} = \sinh(\eta) \tilde e_{\dot\a}^{\mu}
\end{align}
with constant $\tilde e_{\dot\a}^{\mu} \ \in\R$ is  reproduced by the pure background
\begin{align}
T_{\dot\a}  &= \tilde e_{\dot\a}^{\mu} t_\mu \ .
 \label{pure-VF-reconstruct-2}
\end{align}
More generally, all such frames $\tensor{E}{_{\dot\a}^\mu}$ \eq{pure-BG-explicit}
are in the kernel of the 
map $\cD^{++}$ \eq{Dpp} which links the components in $\cC^0$ and $\cC^2$,
and they are  closed under $SO(4,1)$ gauge transformations.
They provide  clean
configurations which can be used as a starting point
for a perturbative approach 
around any point $\xi\in\cM^{3,1}$.

\subsubsection{Generic backgrounds and local linearization}

In general, the reconstructed frame $E^{\dot\a} = \{T^{\dot\a},.\}$ will contain components in $\cC^2$, which are obtained from the $\cC^0$ components via $\cD^{++}$.
However as shown above, we can 
reproduce any given classical frame 
at some fixed point $p$
by a pure frame, which is in the 
kernel of $\cD^{++}$. Writing the classical 
frame in the form 
\begin{align}
 e_{\dot\a}^{\mu} = 
 \bar e_{\dot\a}^{\mu} + \d e_{\dot\a}^{\mu}
\end{align}
where $\bar e_{\dot\a}^{\mu}$ is pure and $\d e_{\dot\a}^{\mu}$ vanishes at $p$,
the reconstructed frame is given by
\begin{align}
 E_{\dot\a}^{\mu} = 
 \bar e_{\dot\a}^{\mu} + (1+\cD^{++})\d e_{\dot\a}^{\mu} \ .
\end{align}
Since $\cD^{++}$ is norm-preserving \eq{Dpp-isometry}, this is well approximated
or dominated by its $\cC^0$ component
 in some sufficiently small neighborhood of $p$. We can make this more quantitative
by recalling that the torsion has the structure 
$T \sim e^{-1}\del e$, whose scale is set by the curvature scale of 
gravity since $\cR \sim T^2$.
Therefore the frame is essentially constant in a regions small
compared to the curvature scale of gravity, and  can be approximated 
by its classical spin 0 
component. Similarly, the  metric is then  well approximated 
by its classical component  $[G_{\mu\nu}]_0  \in \cC^0$.
This justifies the above reconstruction procedure for the frame.
Since the torsion is given by derivatives of the frame and 
observing  \eq{Dpp-deriv-approx},
the intertwiner $\cD^{++}$
applies also to the torsion
\begin{align}
 \tensor{T}{_{\dot\a}_{\dot\b}^\mu} 
 \approx (1+\cD^{++})\big[\tensor{T}{_{\dot\a}_{\dot\b}^\mu}\big]_0 \ .
\end{align}
In particular, a gravitational action of the form
\eq{T-T-action} reduces to the classical $\cC^0$ contribution
due to \eq{Dpp-isometry}.
We can therefore consider the geometrical tensors
as effectively classical in the weak gravity regime,
where the curvature scale $\cR \sim T^2$ is much smaller than all other physical 
scales. 

In the strong gravity regime, these  arguments are no longer justfied.
Nevertheless, we will show  in the  next section that one can always choose local normal coordinates at any given point $p\in\cM^{3,1}$
such that all $\hs$ components of the frame  vanish at $p$.
In this sense the metric always reduces to a classical metric,
which governs the local physics near $p$.

\section{Realization of generic $3+1$-dimensional geometries in matrix models}
\label{sec:generic-4D}

Finally, we address the question if any given metric $G_{\mu\nu}$ 
can be realized in terms of a divergence-free frame. 
The first step is to determine the dilaton, which is obtained from 
\eq{rho-metric-general-cov} as
\begin{align}
 \r^2 = \r_M^{-1}\sqrt{|G|} \ .
\end{align}
The next step is to find some classical divergence-free frame $\tensor{e}{^{\dot a}^\mu}$
which gives rise to \eq{metric-gamma}
\begin{align}
\r^2 G^{\mu\nu} = 
 \g^{\mu\nu} =  \eta_{{{\dot a}}{\dot b}} \tensor{e}{^{\dot a}^\mu} \tensor{e}{^{\dot b}^\nu} \ .
 \label{gamma-frame-e}
\end{align}
Without the constraint, there are of course many frames
(in fact a 6-dimensional orbit of $SO(3,1)$) which achieve that.
The 4 divergence constraints are fairly easy to take into account in Cartesian coordinates $x^\mu$:
for any given space-like components 
$\tensor{e}{^{\dot a}^j}$, the time components $\tensor{e}{^{\dot a}^0}$ are  determined   by
\begin{align}
  \del_0(\r_M  \tensor{e}{^{\dot a}^0}) = - \del_j(\r_M \tensor{e}{^{\dot a}^j}) \ .
\label{div-free-rM}
\end{align}
This can be viewed as an ordinary differential equation in $x^0$, which is solved by
 \begin{align}
   \tensor{e}{^{\dot a}^0} = -\r_M^{-1}\,\int\limits_{\xi_0}^{x^0} d\xi\, \del_j(\r_M \tensor{e}{^{\dot a}^j}) \  + \tensor{e}{^{\dot a}^0}(\xi_0)
   \label{frame-0-solve}
 \end{align}
 where the value $\tensor{e}{^{\dot a}^0}(\xi_0)$ at any given time $\xi_0$ can be chosen as desired.
 This means that we  can freely choose the 12 space-like $\tensor{e}{^{\dot a}^j}$, which should allow to reproduce the 10 dof in $\g^{\mu\nu}$ even if the divergence constraint is imposed.

A more systematic, iterative way to determine the frame is as follows:
choose some reference point $\bar x$. After a global $SO(3,1)$ transformation on the frame indices,
we can assume that $\g^{\mu\nu}|_{\bar x} = c\eta^{\mu\nu}$, and we assume $c=1$ for simplicity.
Then choose the diagonal elements as $\tensor{e}{^{\dot a}^a} = \eta^{\dot a a}$, and  off-diagonal frame elements which vanish at $\bar x$, such that the frame reproduces $\g^{\mu\nu}$.
To satisfy the divergence constraint, we define a correction of the diagonal frame elements by
\begin{align}
  \d\tensor{e}{^{\dot a}^a} =  - \r_M^{-1}\,\int\limits_{\bar x^a}^{x^a} d\xi^a\, \sum\limits_{\mu\neq \dot a}\del_\mu(\r_M \tensor{e}{^{\dot a}^\mu}) \ ,
  \label{frame-diag-solve}
\end{align}
which vanishes at $\bar x$. Then the improved frame $\tensor{e}{^{\dot a}^a} \to \tensor{e}{^{\dot a}^a} +  \d\tensor{e}{^{\dot a}^a}$ satisfies the divergence constraint, and reproduces $\g^{\mu\nu}$
to a good approximation near $\bar x$. Now we repeat this procedure iteratively by correcting the
off-diagonal elements of the frame such that $\g^{\mu\nu}$ is reproduced,
and correcting the diagonal elements again with
\eq{frame-diag-solve}, and so on. Since the corrections vanish at $\bar x$,
this procedure will converge to a divergence-free frame
which reproduces $\g^{\mu\nu}$ exactly at least in some neighborhood of $\bar x$.
This could presumably
be proved e.g. using the Banach fixed point theorem, but we leave it as a plausibility argument here
and accept the statement as true.

We conclude that there are always divergence-free frames $\tensor{e}{^{\dot a}^\mu}$
which realize \eq{gamma-frame-e} for any $\g^{\mu\nu}$.
As explained in section \ref{sec:reconstruction-M31}, we can then find a corresponding
matrix background which implements the frame in the weak gravity regime. 
Therefore generic $3+1$-dimensional space-time geometries can indeed
be implemented as backgrounds of the matrix model with an ansatz of the form \eq{general-BG}, leading to a covariant quantum space-time.

Moreover, the above analysis shows that 2 of the 12 dof in $\tensor{e}{^{\dot a}^\mu}$
remain undetermined even if the divergence constraint is imposed.
They can be used to restrict the totally
antisymmetric components 
of the torsion \eq{T-tilde-dual}, which define a vector field via
$\tilde T_\kappa \propto \tensor{T}{^{(AS)}^\nu^\s^\mu}\varepsilon_{\nu\s\mu\kappa}$.
For example, it is plausible that the frame can be chosen such that
\begin{align}
 \tilde T_\mu = \psi^{-1}\del_\mu \tilde\r \
 \label{T-del-onshell}
\end{align}
in terms of an axion $\tilde\r$; this is a consequence of the (semi-classical)
matrix model equations of motion \cite{Fredenhagen:2021bnw}.
This question and its implications should be addressed elsewhere.

\section{Quantization, extra dimensions and induced gravity}

Even though the semi-classical matrix model action defines a
dynamical theory of space-time geometry, it is expected that 
 a (near-) realistic theory of gravity can be obtained only 
from the Einstein-Hilbert action. Remarkably, this arises indeed 
in the 1-loop effective action under certain assumptions,
in the spirit of induced gravity  \cite{Sakharov:1967pk,Visser:2002ew}.
The quantization of the matrix model is defined non-perturbatively through a 
matrix path integral
 \begin{align}
  \cZ = \int dT d\Psi e^{i S} \ .
  \nn
\end{align}
The oscillatory integral becomes absolutely convergent for finite-dimensional $\cH$ upon implementing the regularization 
\begin{align}
 S\to S + i\varepsilon \sum_{\dot\a} Y_{{\dot\a}} Y_{{\dot\a}} \ ,
 \label{Feynman}
\end{align}
which amounts to a Feynman $i\varepsilon$ term in the noncommutative gauge theory.
For recent results of numerical simulations of such models see e.g. 
\cite{Nishimura:2019qal,Anagnostopoulos:2020xai}.

In general, 
the quantization of matrix models on some noncommutative background leads to highly non-local action
due to UV/IR mixing, {\em except} in the maximally supersymmetric IKKT model.
This phenomenon was shown first identified in \cite{Minwalla:1999px}, but it is most transparent in terms of 
string states $|x\rangle\langle y| \in \End(\cH)$, which govern the 
deep quantum (or extreme UV) regime of noncommutative functions
\cite{Steinacker:2022kji,Steinacker:2016nsc}.
These states are also extremely useful to compute the 
1-loop effective action of the IKKT matrix model on generic backgrounds.
It was indeed show in \cite{Steinacker:2021yxt}  that the Einstein-Hilbert action arises at 1 loop,
{\em provided} the transversal 6 matrices $T^{\dot a}$
of the IKKT model assume some
non-trivial background given by some compact fuzzy space:
\begin{align}
 T^{\dot k} \sim t^{\dot k} : \quad\cK \hookrightarrow \R^{6}, \qquad \dot k =4,...,9 \ .
\end{align}
This describes a quantized compact symplectic space
$\cK$ embedded along the transversal directions, which plays
the role of fuzzy extra dimensions. Together with the space-time 
brane $\cM^{3,1}$, the overall background geometry then has a product structure
\begin{align}
 \cM^{3,1} \times \cK  \ \ \hookrightarrow \R^{9,1} \ .
 \label{product-ansatz}
\end{align}
The detailed structure of $\cK$ will be irrelevant\footnote{$\cK$ could be a fuzzy sphere $S^2_N$, or some richer fuzzy space leading to interesting low-energy gauge theories, cf.
\cite{Chatzistavrakidis:2011gs}.};
we only require that the internal 
matrix Laplacian $\Box_\cK = [T^{\dot k},[T_{\dot k},.]]$ has positive spectrum, 
\begin{align}
 \Box_\cK Y_{\L} = m^2_\L\, Y_{\L} \ , \qquad m_{\L}^2 = m_\cK^2 \mu^2_{\L} \
 \label{KK-masses}
\end{align}
with a finite number of (Kaluza-Klein KK) eigenmodes
$Y_{\L} \in \End(\cH_\cK)$ enumerated by some  label $\L$. Here $m_\cK^2$
determines the radius of $\cK$ and
sets the  scale of the KK modes, which will play an important role below.

Computing the 1-loop effective action on such a background then leads in particular to
the following term 
\cite{Steinacker:2021yxt}
\begin{align}
 \Gamma_{\rm 1 loop}^{\cK-\cM}  
 = -\frac{c_{\cK}^2}{(2\pi)^4} \int\limits_{\cM}d^4 x\sqrt{G}\, \r^{-2} m_\cK^2
  \tensor{T}{^\r_\s_{\mu}}\tensor{T}{_{\r}^{\s}_{\nu}} G^{\mu\nu} 
  \label{T-T-action}
\end{align}
which describes the effective interaction between $\cM^{3,1}$ and $\cK$.
Here 
 \begin{align}
 c_{\cK}^2 = \frac{\pi^2}{8}\sum_{\L,s}
 \frac{(2s+1) C^2_{\L}}{\mu_{\L}^2 + \frac{m_s^2}{m^2_\cK}}\  > 0 \
 \label{C2-K-cutoff}
\end{align}
is finite, determined by the dimensionless KK masses $\mu_{\L}$ on $\cK$ \eq{KK-masses}
and their cousins $C^2_{\L}$, which also depend on the  structure of $\cK$.
The mass scale of the internal $\hs$ modes on $S^2_n$ is given by
\begin{align}
 m_s^2 = \frac{s(s-1)}{R^2} \ .
\end{align}
Using partial integration,
one can rewrite the above effective action
in terms of an Einstein-Hilbert term with effective Newton constant
\begin{align}
  \frac{1}{16\pi G_N} = \frac{c^2_{\cK}}{14\pi^4} \r^{-2} m_\cK^2\, .
  \label{Newton-constant-rho-mK}
 \end{align}
 However, this requires  assuming some specific
behavior of $m^2_\cK$ or $G_N$.
If we assume   $G_N = const$, we can use the identity \cite{Steinacker:2021yxt}
  \begin{align}
 \int d^{4}x \frac{\sqrt{|G|}}{G_N}\, \cR
 &= -\int d^{4}x\frac{\sqrt{|G|}}{G_N}\,\Big(
 \frac 78\tensor{T}{^\mu_\s_\r} \tensor{T}{_\mu_{\s'}^\r} G^{\s\s'}
 + \frac 34\tilde T_{\nu} \tilde T_{\mu}  G^{\mu\nu}   \Big) \
 \label{S-EH-T-id-GN-const}
\end{align}
where $\cR$ is the Ricci scalar of the effective metric $G_{\mu\nu}$, and
\begin{align}
 \tilde T_{\mu} dx^\mu = -\star(\frac 12 G_{\nu \s}\tensor{T}{^{\s}_\r_\mu}dx^\nu dx^\r dx^\mu)
 \label{T-tilde-dual}
\end{align}
 is the Hodge-dual
of the totally antisymmetric torsion.
This gives
\begin{align}
\boxed{\
 \Gamma_{\rm 1 loop}^{\cK-\cM}
  = \int\limits_\cM d^4 x  \,\frac{\sqrt{G}}{16\pi G_N}
  \Big(\cR + \frac 34 \tilde T_{\nu} \tilde T_{\mu}  G^{\mu\nu}  \Big) \ .
    \  }
   \label{Gamma-EH-I-2-cov}
\end{align}
  Using the
 eom of the matrix model, $\tilde T_{\nu}$ reduces to a gravitational axion $\tilde\r$  \cite{Fredenhagen:2021bnw}
\begin{align}
 \tilde T_\mu = \r^{-2}\del_\mu\tilde\r \ .
\end{align}
Since $\tilde T_\mu$ vanishes exactly on the cosmic background, it is plausible 
that its effect is small, in which case we  recover the Einstein-Hilbert action
as desired.

However since $G_N$ depends on $\r$ and $m_\cK$, it is not evident that $G_N = const$. If we assume instead that  $m_\cK = const$ (which is reasonable as discussed below),
then one can derive an analogous identity
\begin{align}
 \int d^{4}x\frac{\sqrt{|G|}}{G_N}\, \cR
 &= -\int d^{4}x\frac{\sqrt{|G|}}{G_N}\,\Big(\frac 18 \tensor{T}{^\mu_\s_\r} \tensor{T}{_\mu_{\s'}^\r} G^{\s\s'}
  + \frac{1}{4}\tilde T_{\nu} \tilde T_{\mu}  G^{\mu\nu}   \Big) \
  \label{S-EH-T-id-mK-const}
\end{align}
 based on results in \cite{Fredenhagen:2021bnw}.
This leads to a slightly modified gravitational action
\begin{align}
\boxed{\
 \Gamma_{\rm 1 loop}^{\cK-\cM}
  = 7\int\limits_\cM d^4 x \frac{\sqrt{G}}{16\pi G_N}
  \Big(\cR + \frac 14 \tilde T_{\nu} \tilde T_{\mu}  G^{\mu\nu} \Big) \ }
   \label{Gamma-EH-I-GN-cov}
\end{align}
where the Newton constant is modified by a factor 7.
The precise form of the gravitational action thus depends on the behavior of the compactification
scale $m^2_\cK$, which needs to be clarified in future work.

These results are remarkable in many ways.
The first observation is that the  Newton constant $G_N$ \eq{Newton-constant-rho-mK} is set by
 the compactification scale $m_\cK$.
 This means that the Planck scale is related to the Kaluza-Klein scale for the 
 fuzzy extra dimensions $\cK$.
 Without the fuzzy extra-dimensional $\cK$, no  Einstein-Hilbert action is induced,
and only some (rather obscure) higher-derivative action is obtained.
It should be noted that no UV divergence arises in the loop computation, due to maximal supersymmetry of the matrix model and the
 fact that $\cK$ supports only a finite number of modes.

We can justify the presence of $\cK$ to some extent by studying how the 
1-loop effective action depends on its radius, or equivalently on $m_\cK$.
This is obtained from the same computation as above:
It turns out that  \eq{T-T-action}
\begin{align}
 \Gamma_{\rm 1 loop}^{\cK-\cM} = c^2 m_\cK^2 = - V_{\rm 1 loop}(m_\cK^2) > 0
\end{align}
 is positive for the covariant
FLRW space-time in \cite{Sperling:2019xar}. 
Combined with the bare matrix model action, the effective potential has the structure 
\begin{align}
 V(m_\cK^2) = - c^2 m_\cK^2 + \frac{d^2}{g^2} m_\cK^4
 \label{V-K}
\end{align}
 at weak coupling. This clearly has a minimum for $m_\cK^2>0$ with $V < 0$.
Since $m_\cK$ is essentially the radius of $\cK$,
this strongly suggests that  $\cK$ is  stabilized by quantum effects, thus
providing some justification for \eq{product-ansatz}.

One may worry that the effective potential for $m_\cK$ depends on the
geometry of $\cM^{3,1}$, which we have assumed to be the cosmic background
brane. Thus gravitational deformations of the geometry should have some influence on the Newton constant.
Nevertheless,  $m_\cK$ is expected to be constant to a very good
approximation. Since $m_\cK$ is essentially the radius of $\cK$,
its kinetic term $\int \del^\mu m_\cK \del_\mu m_\cK$ in the matrix model is huge,
which would strongly suppress any local variations;
note that $m^2_\cK \sim \r^2 G_N^{-1}$
is a huge energy   beyond the Planck scale.
Therefore $m_\cK$ should be almost constant, and hence
governed by the large-scale cosmic background
as assumed above.

On the other hand, this suggests that the Newton constant may change
during the cosmic expansion. This may be a significant concern, since there
are rather strong observational bounds on such a variation. Nevertheless, at this
early stage such worries are presumably sub-leading,
and the prime focus should be to gain a more detailed understanding
of this new mechanism for gravity.

Furthermore,
 the above induced gravity {\em action} in 3+1 dimensions can be interpreted as a quasi-local {\em interaction}
of $\cK$ and $\cM$ via 9+1-dimensional IIB supergravity,
recalling that the 1-loop effective action is related to IIB supergravity \cite{Ishibashi:1996xs,Chepelev:1997av,Taylor:1998tv,Steinacker:2016nsc}. This provides additional
confidence into the above rather formal computations, since $9+1$-dimensional
supergravity is well established  in string theory and expected to be recovered
in the matrix model. A more detailed understanding of the relation with supergravity
for backgrounds of the structure $\cM^{3,1 }\times\cK\subset\R^{9,1}$ would be desirable.

Note that in contrast to orthodox string theory, target space $\R^{9,1}$ is not
compactified here.
This makes sense, since the  perturbative physics on such backgrounds is restricted to the brane, and there are no bulk modes radiating off the brane  at weak coupling.
Hence the main problem of  string theory - i.e.
the need for compactification and the
lack of preferred choices thereof - turns into a blessing, as there would be
no induced gravity on space-time without the extra dimensions of target space.

\paragraph{Vacuum energy due to $\cK$.}

The 1-loop contribution to the vacuum energy due to 
$\cK$ is obtained using an analogous trace computation, 
leading to a result of the structure 
\begin{align}
\Gamma_{\rm 1 loop}^\cK
 = \frac{3i}4\Tr\Big(\frac{V_4^\cK}{\Box^4}\Big)
  &\sim -\frac{\pi^2 }{8(2\pi)^4}
  \int\limits_\cM \!\Omega\, \r^{-2} m_\cK^4\sum_{\L s}
  \frac{V_{4,\L}}{\mu^4_{\L}}  
  \label{vacuum-energy-K}
\end{align}
assuming $\frac 1{R^2} \ll m^2_{\L}\sim m_\cK^2$. 
Here $V_{4,\L}$ depends on the structure of $\cK$. This is typically a large vacuum energy with scale set by $m_\cK$ which was related to the Planck scale above,
which could  have either sign.
However as the symplectic volume form  $\Omega$ 
is independent of the metric,
this 1-loop vacuum energy is {\em not} equivalent to a cosmological constant;
its effect on the dilaton $\r$ remains to be understood.
The present framework can therefore be viewed as a realization of induced
gravity in the spirit of Sakharov \cite{Sakharov:1967pk,Visser:2002ew}, 
which is free of UV divergences, and appears to avoid the associated cosmological constant problem.

\subsection*{Acknowledgments}

Useful discussions with Y. Asano, S. Fredenhagen, M. Hanada, V.P. Nair and J. Tekel
are gratefully acknowledged.
The author would  like to thank the organisers of the
Corfu Summer Institute 2021 and the Humboldt Kolleg on “Quantum Gravity and Fundamental
Interactions” for the stimulating meeting and the invitation to deliver a talk.
This work was supported by the Austrian Science Fund (FWF), project P32086.

\end{document}